\date{}
\documentclass[11pt]{article}
\usepackage{amsmath,amsthm,amsfonts,amssymb,amsopn,amscd}

\def\qed{{\unskip\nobreak\hfil\penalty50
\hskip2em\hbox{}\nobreak\hfil$\square$
\parfillskip=0pt \finalhyphendemerits=0\par}\medskip}
\def\proof{\trivlist \item[\hskip \labelsep{\bf Proof.\ }]}
\def\endproof{\null\hfill\qed\endtrivlist\noindent}

\def\Ad{{\hbox{\rm Ad}}}

\def\l{\lambda}

\def\A{{\cal A}}
\def\B{{\cal B}}
\def\C{{\cal C}}
\def\D{{\cal D}}
\def\E{{\cal E}}

\def\I{{\cal I}}
\def\H{{\cal H}}
\def\K{{\cal K}}
\def\S{{\cal S}}

\def\f{{\varphi}}

\def\Si{\mathbb S_{\infty}}

\def\S2{S^{1(2)}}

\def\Ed{{\cal E_{\delta}}}

\newtheorem{theorem}{Theorem}[section]
\newtheorem{lemma}[theorem]{Lemma}

\newtheorem{corollary}[theorem]{Corollary}

\newtheorem{proposition}[theorem]{Proposition}

\theoremstyle{definition} 

\theoremstyle{remark}

\def\Hin{\mathbb H^{\infty}(\mathbb S_{\infty})}

\def\RR{{\mathbb R}}
\def\CC{{\mathbb C}}

\def\sl2{{{\rm SL}(2,\RR)}}
\def\psl2{{{\rm PSL}(2,\RR)}}

\def\u1{{{\rm V}(1)}}
\def\su2{{{\rm SV}(2)}}

\def\so3{{{\rm SO}(3)}}

\def\A{{\mathcal A}}
\def\B{{\mathcal B}}
\def\C{{\mathcal C}}
\def\D{{\mathcal D}}

\def\H{{\mathcal H}}
\def\I{{\mathcal I}}
\def\K{{\mathcal K}}

\def\O{{\mathcal O}}

\def\LH{{\mathfrak H_R}}

\newcommand{\be}{\begin{equation}} 
\newcommand{\ee}{\end{equation}}
\newcommand{\bea}{\begin{eqnarray}} 
\newcommand{\ea}{\end{eqnarray}}

\newcommand{\inv}{^{-1}}

\usepackage{epsfig}
\usepackage{amsmath, amssymb}
\begin{document}

\title{\huge Boundary Quantum Field Theory\\
on the Interior of the Lorentz Hyperboloid}
 
\author{{\sc Roberto Longo}
\\
Dipartimento di Matematica,
Universit\`a di Roma ``Tor Vergata'',\\
Via della Ricerca Scientifica, 1, I-00133 Roma, Italy
\\
\phantom{X}\\
{\sc Karl-Henning Rehren} \\
Institut f\"ur Theoretische Physik, Universit\"at G\"ottingen,
\\ Friedrich-Hund-Platz 1, D-37077 G\"ottingen, Germany}

\maketitle

\begin{abstract} We construct local, boost covariant boundary QFT nets
  of von~Neumann algebras on the interior of the Lorentz hyperboloid
  $\LH$, $x^2 - t^2 > R^2$, $x>0$, in the two-dimensional Minkowski
  spacetime. Our first construction is canonical, starting with a local
  conformal net on $\RR$, and is analogous to our previous
  construction of local boundary CFT nets on the Minkowski
  half-space. This net is in a thermal state at Hawking
  temperature. Then, inspired by a recent construction by E. Witten
  and one of us, we consider a unitary semigroup that we use to build
  up infinitely many nets. Surprisingly, the one-particle semigroup is
  again isomorphic to the semigroup of symmetric inner functions of
  the disk. In particular, by considering the $U(1)$-current net, we
  can associate with any given symmetric inner function a local, boundary QFT net on $\LH$. By considering different states, we shall also have nets in a ground state, rather than in a KMS state.
\end{abstract}

\vskip5cm

\noindent
----------------

\noindent
{\footnotesize Supported by the ERC Advanced Grant 227458  
OACFT ``Operator Algebras and Conformal Field Theory", PRIN-MIUR,
GNAMPA-INDAM and EU network ``Noncommutative Geometry"
MRTN-CT-2006-0031962. Supported in part by the German
    Research Foundation (Deutsche Forschungsgemeinschaft (DFG))
    through the Institutional Strategy of the University of
    G\"ottingen.

\noindent
E-mail: {\tt longo@mat.uniroma2.it, rehren@theorie.physik.uni-goettingen.de}}

\newpage

\section{Introduction}
An algebraic description of Boundary CFT on the two-dimensional
Minkowski half-space $x>0$ has been given in \cite{LR1}. Recently an
infinite new family of Boundary QFT on the half-space has been set up
in \cite{LW} by modifying this construction by elements of a unitary
semigroup, a construction conceptually related to the inverse
scattering method, in the sense that our first models are associated
with (symmetric) scattering functions. 
 
In two dimensions, since locally every Lorentz geometry is conformally
flat, boundaries in Minkowski spacetime may be regarded as an analogue
of curvature and horizons in higher dimensions. Considering QFT with and
without boundaries therefore is a testing ground for the generally
covariant locality principle \cite{BFV}, a new paradigm for local
quantum field theory stipulating a simultaneous definition on flat and
curved spacetimes.  

As pointed out to us by Edward Witten, the same question arises in
open string theory, notably in the context of gauge-gravity
(open-closed) duality. In his words, ``studying all possible
extensions to manifolds with boundary of a specific 2d theory defined
initially on manifolds without boundary is the analog of studying
Yang-Mills theory in the background of a specified solution of
Einstein's equations.'' His suggestion that this is a question the
algebraic approach to QFT might shed light on, had motivated our
previous work \cite{LW}. 


In this paper we then study local, relativistic boundary CFT and QFT 
on the interior $\LH$ of the Lorentz hyperboloid $x^2 - t^2= R^2$ ($x>0$), 
in the Minkowski plane. For simplicity we put $R=1$ in this section. 

A two-dimensional CFT on $\LH$ is a local QFT with a conserved and
traceless stress-energy tensor, subject to a boundary condition at
the boundary $x^2 - t^2=1$. As is well known,  
conservation and vanishing of the trace imply that the components 
$T_L=\frac12(T_{00}+ T_{01})$ and $T_R=\frac12(T_{00}- T_{01})$ are
chiral fields, $T_L(u)$, $T_R(-v)$, where we are using light-cone
coordinates $u= x+t$, $v=x-t$. The boundary condition is the vanishing
of energy flux across the boundary, $T^{0\mu}\epsilon_{\mu\nu}dx^\nu =
0$, which in components becomes $(t+x) T_L(t+x) = (t-x) T_R(t-x)$,
namely $uT_L(u)|_{uv=1} =  -vT_R(-v)|_{uv=1}$, so
\[
uT_L(u) =- \frac{1}{u}T_R\left( -\frac{1}{u}\right) \equiv T(u)\ .
\]
It follows that the components $T_{10}=T_{01}$, $T_{11}=T_{00}$ of the
stress-energy tensor are of the form 
\begin{equation} 
T_{00}(u,v) =  \frac{1}{u}T(u) - vT\left(- \frac{1}{v}\right), \qquad T_{01}(u,v) =  \frac{1}{u}T(u) + vT\left(- \frac{1}{v}\right), 
\end{equation}
i.e., bi-local expressions in terms of the chiral field $T$. 

In terms of the local von~Neumann algebras $\A(\O)$ generated by the stress-energy tensor, this means that if $I,J$ are bounded intervals of $\RR^+$ with $\O= I\times J$ a double cone contained in $\LH$ (thus $uv>1;\, u\in I, v\in J$) we have
\begin{equation}\label{main}
\A(\O) = \A_0(I)\vee \A_0(J^{-1}) \ ,
\end{equation}
where $\A_0$ is the net on $\RR$ generated by the chiral stress-energy tensor (Virasoro net).
Indeed only the restriction $\A_0|_{\RR^+}$ of $\A_0$ to the positive half-line enters here.

Dilation covariance of $\A_0|_{\RR^+}$ gives boost covariance of $\A$ and the KMS property of the vacuum state on $\A_0(\RR^+)$ (Bisognano-Wichmann property) gives the KMS property of the vacuum state on $\A(\LH)$ w.r.t.\ the boosts at Hawking-Unruh inverse temperature $\beta= 2\pi$.

In more generality, starting with any dilation covariant local net $\A_0$ of von~Neumann algebras on $\RR^+$ we may associate by formula \eqref{main} a local boost covariant net of von~Neumann algebras on $\LH$. Of course, we could also start with a local translation covariant net on $\RR$ as there is a one-to-one correspondence between local translation covariant nets on $\RR$ and local dilation covariant nets on $\RR^+$ by the change of variable $x\leftrightarrow e^x$. 

Now we can extend the above canonical construction \eqref{main} based on two local dilation covariant nets $\A_0$ and $\B_0$ on $\RR^+$ such that $\B_0$ is forwardly local w.r.t.\ $\B_0$, i.e., $\A_0(I_1)$ commutes with $\B_0(I_2)$ if $I_2 > I_1$:
\[
\A(\O) = \B_0(I)\vee \A_0(J^{-1}) \ ,
\]
where $J^{-1} = \{1/v: v\in J\}$. The option to have different $\A_0$
and $\B_0$ is interesting not least in string theory, because the ``heterotic
string'' is of that kind (different internal symmetry groups for
left- and right-moving modes, but equal central charge; hence $\A_0$
and $\B_0$ share a common subnet of the stress-energy tensor). they
are both conformal, but not related by a unitary. 

However, in analogy with the model building in \cite{LW}, we may start
with a local dilation covariant net $\A_0$ and consider the semigroup
$\E_\delta(\A_0)$ of unitaries $V$ commuting with dilations such that 
\[
V\A_0(1,\infty)V^*\subset \A_0(1,\infty)\ ,
\]
Thus $V\A_0(a,\infty)V^*\subset \A_0(a,\infty)$ for all $a>0$.
Every element $V\in\E_\delta(\A_0)$ then gives a dilation covariant
net $\B_0 = V\A_0 V^*$ that is forwardly local w.r.t.\ $\A_0$, thus a
local, boost covariant, Boundary QFT net of von~Neumann algebras on
$\LH$ by the formula 
\[
\A(\O) = V\A_0(I)V^*\vee \A_0(J^{-1}) \ .
\]
At this point we may produce families of local, boost covariant nets
on $\LH$ once we compute non-trivial elements of $\E_\delta(\A_0)$ for
a given local, M\"obius covariant net on $\RR$.  

We shall study the semigroup $\E_\delta(\A^{0})$ with $\A^{0}$ the
$U(1)$-current net on $\RR$. One of our main result is the computation
of the sub-semigroup of $\E_\delta(\A^{0})$ consisting of second
quantization unitaries (unitaries $V$ that are promotions of
one-particle unitaries $V_0$); we have 
\[
V =\Gamma(V_0)\in\E_\delta(\A^{0}) \Leftrightarrow V_0 = \f(K),\ \f\
\text{symmetric inner function on the upper half-plane} 
\]
where $K$ is the one-particle generator of the dilation semigroup.

Thus, rather surprisingly, the above sub-semigroup is naturally
isomorphic to the semigroup of inner functions as was the case for the
Boundary QFT on the half-space, where positivity of the energy played
a crucial role; that role is here played by the KMS thermal
equilibrium property. 

\section{Basic definitions}\label{basicdef} 

Let $M$ be the two-dimensional Minkowski spacetime and fix $R>0$. We shall consider the spacetime $\LH=\{(t,x)\in M: x^2 - t^2 > R^2,\, x>0\}$, the interior of the Lorentz hyperboloid  $x^2 - t^2 = R^2,\, x>0$. 
Then $\LH$ inherits the Lorentz metric from $M$.

The Lorentz boosts provide a one-parameter group of diffeomorphisms $\Lambda$ of $\LH$: 
\[
\Lambda(s) =
\left(\begin{array}{cc}\cosh s & \sinh s \\ \sinh s &\cosh s\end{array}\right)\ .
\]
In light-cone coordinates $u = x+ t$, $v= x-t$, the Lorentz hyperboloid is given by $uv=R^2$ and $\LH$ is the region $uv>R^2$, $u>0$.

We shall denote by $\K$ the set of double cones strictly contained in $\LH$, namely $\O\in \K$ if $\O = I_1\times I_2$ with $I_1$ and $I_2$ bounded intervals of the chiral lines $u=0$, $v=0$ and $\bar\O\subset \LH$: so  $\O= \{(u,v): u\in I_1, v\in I_2, uv>R^2\}$ and $uv>R^2$ for all $u\in\bar I_1$, $v\in \bar I_2$.

A local \emph{net $\A$ of  von~Neumann algebras on $\LH$} is a map
\[
\O\mapsto\A(\O)
\]
from $\K$ to the set of von~Neumann algebras on a (fixed) Hilbert space $\H$ that satisfies the following properties:
\begin{description}
\item[$\textnormal{\textsc{1. Isotony}}$:] {\it If $\O_1$, $\O_2$ are double cones and $\O_1\subset \O_2$, then $\A(\O_1)\subset\A(\O_2)$.}
\item[\textnormal{\textsc{2. Boost invariance}}:] {\it There is a
strongly continuous one-pameter unitary group $U$ on $\H$ such that
$U(s)\A(\O)U(s)^*=\A(\Lambda(s)\O)$, $\ t\in \RR$,$\ \O\in\K$.}
\item[$\textnormal{\textsc{3. Locality}}$:]
{\it If $\O_1, \O_2\in\K$ are spacelike separated, the von~Neumann algebras $\A(\O_1)$ and $\A(\O_2)$ commute.}
\item[\textnormal{\textsc{4. Invariant state}}:]
{\it There exists a unit $U$-invariant vector $\xi$, cyclic for $\A(\LH)\equiv\bigvee_{\O\in\K}\A(\O)$}.
\end{description}
If the generator of the one-parameter group is positive, we shall say
that $\xi$ is a \emph{vacuum vector}, the state $(\xi,\cdot\,\xi)$ a vacuum (equivalently: ground) state, and that the net is in a vacuum (equivalently: ground) representation.

If the one-parameter automorphism group $\tau =\Ad U$ of $\A(\LH)$  satisfies the KMS condition w.r.t.\ $(\xi,\cdot\,\xi)$, we shall say that $\A$ is in a KMS representation.

\subsection{Translation and dilation covariant nets on $\RR$}
\emph{A local, translation (resp.\ dilation) covariant net of von~Neumann algebras on $\RR$ (resp.\ on $\RR^+$)} on a Hilbert space $\H$ is a triple $(\A_0, U, \xi)$ where
\begin{itemize}
\item $\A$ is an isotonous map 
\[
I\mapsto \A_0(I)
\] 
where $I\in\I\ (\text{resp.}\ I\in\I_+)$ and $\A_0(I)$ is a von~Neumann algebra on $\H$;
\item
$U$ is a one-parameter group on $\H$ such that
\[
U(t)\A_0(I)U(-t) = \A_0(I+ t),\ (\text{resp.}\ \A_0(e^t I)),\quad \forall I\in\I \ (\text{resp.} \ \I_+) ,\ t\in\RR;
\]
\item
$\xi\in\H$ is a unit $U$-invariant vector, cyclic for $\bigcup_{I} \A_0(I)$.

\item $\A_0(I_1)$ and $\A_0(I_2)$ commute if $I_1,I_2$ are disjoint intervals.
\end{itemize}
Here $\I$ (resp $\I_+$) is the family of bounded, open, non-empty intervals of $\RR$ (resp of $\RR^+$ with $0\notin I$).
We do not assume positivity of the energy nor irreducibility of the net.

If the last condition is not satisfied, the net is called {\em nonlocal}. 

Note that there is a one-to-one correspondence
\begin{equation}\label{diltrans}
\begin{gathered}
 \text{translation covariant nets on $\RR$}\\
 \updownarrow\\
\text{dilation covariant nets on $\RR^+$}
\end{gathered}
\end{equation}
 simply by the ``change of variable'' $x\leftrightarrow e^x$.

We shall say that the translation covariant nets $\A$ and $\A_1$ are abstractly isomorphic (resp.\ isomorphic) if there is a coherent family of isomorphisms $\Phi_I : \A_1(I)\to\A(I)$, $I\in \I$, interchanging the translation action (and the invariant state). Analogous notions can be given for dilation covariant nets.

Let $\A$ be a local M\"obius covariant net of von~Neumann algebras on $\RR$.
Denote by $U$ (resp.\ $V$) the one-parameter unitary translation (resp.\ dilation) group on $\H$. Then $\A$ is a translation covariant net on $\RR$, and the restriction of $\A$ to $\RR^+$ is a local dilation covariant net on $\RR^+$ (w.r.t.\ the vacuum vector).
We recall the following fact from \cite{CLTW1}:
\begin{proposition}\label{CLTW}
If $\A_0$ is a diffeomorphism covariant local net on $\RR$ and $\A_1$ the translation covariant net on $\RR$ associated with $\A_0|_{\RR^+}$ by the above correspondence, then $\A_1$ and $\A_0$ are abstractly isomorphic as translation covariant nets.  
\end{proposition}
\noindent
Therefore one can carry the vacuum state for $\A_0|_{\RR^+}$ to a translation invariant state for $\A_0$. This is the \emph{geometric KMS state}, a canonical KMS state for $\A_0$.

\subsection{Forwardly local chiral nets}
Let $\A_0$ and $\B_0$ be local nets of von~Neumann algebras on $\RR$ on the same Hilbert space. We shall also assume that both $\A_0$ and $\B_0$ are covariant w.r.t.\ the same one-parameter translation group. 
We shall say that \emph{$\B_0$ is forwardly local w.r.t.\ $\A_0$} if
$\A_0(I_1)$ commutes with $\B_0(I_2)$ for all intervals $I_1 , I_2$ of
$\RR$ such that $I_2 > I_1$ ($I_2$ is in the future of $I_1$). 
We define {\em backwardly local} in the obvious way
($I_2<I_1$). $\B_0$ is {\em relatively local} w.r.t.\ $\A_0$ if it is
both forwardly and backwardly local. 

We shall say that duality for half-lines holds for $\A_0$ if 
\[
\A_0(I)'\cap\A_0(\RR)=\A(I')\ ,
\]
where $I\subset\RR$ is any half-line. The following is immediate.
\begin{lemma}
If $\B_0$ is forwardly local w.r.t.\ $\A_0$ and duality for half-lines holds for $\A_0$, then $\B_0(a,\infty)\subset\A_0(a,\infty)$, $a\in\RR$.
\end{lemma}
\proof 
$\B_0(a,\infty)$ commutes with $\A_0(-\infty, a)$, so is contained in $\A_0(-\infty, a)'\cap\A_0(\RR) = \A_0(a,\infty)$.
\endproof
If duality for half-lines holds for both $\A_0$ and $\B_0$, and $\A_0$ and $\B_0$ are relatively local, then $\B_0(I)=\A_0(I)$ for all half-lines $I\subset\RR$.

It is immediate to translate the above notion for local nets on $\RR^+$ by the correspondence \eqref{diltrans}. If $\A_0$ is a local M\"obius covariant net of von~Neumann algebras on $\RR$, then duality for half-lines holds for $\A_0|_{\RR^+}$ (namely
$\A_0(0,a)'\cap\A_0(0,\infty) = \A_0(a,\infty)$, $a>0$) iff $\A_0$ is strongly additive.

\section{BQFT on the interior of the Lorentz hyperboloid}\label{constr1}
Let $\A_0$ and $\B_0$ be dilation covariant, local nets of von~Neumann algebras on $\RR^+$ with $\B_0$ forwardly local w.r.t.\ $\A_0$. 

Given intervals $I_1, I_2$ of $\RR^+$ we define the von~Neumann algebra $\A(\O)$ associated with the double cone $\O= I_1\times I_2$ by
\begin{equation} \label{bqft}
\A(\O) = \B_0(I_1)\vee \A_0(R^2 I^{-1}_2)
\end{equation}
where $I^{-1} \equiv \{ \l^{-1} : \l\in I\}$.

Formally we can write
\[
\A(u,v) = \B_0(u)\vee \A_0(R^2/v)
\]
and for simplicity we use this formal writing here below; one can easily properly write up by replacing a point $(u,v)\in\LH$ with a double cone $\O\in\K$.
\begin{theorem} 
$\A$ is a local, boost covariant net of von~Neumann algebras on $\LH$. 
\end{theorem}
\proof
To check locality, we have to show that if $(u,v)$ and $(u',v')$ are points of $\LH$ that are spacelike, i.e., $(u' - u)(v' - v) >0$,  then $\A(u,v)$ and $\A(u',v')$ commute. So choose $(u,v)$ and $(u',v')$ with $uv> R^2$ and $u'>u$, $v'>v$. 

Since $u'>u> R^2/v > R^2/v'$, $\B_0(u')$ commutes with $\B_0(u)$ (by locality of $\B_0$) and with $\A_0(R^2/v)$ (by forward locality), so $\B_0(u')$ commutes with $\A(u,v)$. Moreover $\A_0(R^2/v')$ commutes with $\A_0(R^2/v)$ (by locality of $\A_0$) and with $\B_0(u)$ (by forward locality). Therefore $\B_0(u')\vee \A_0(R^2/v')$ commutes with $\B_0(u)\vee \A_0(R^2/v)$.

Concerning the covariance, let $U$ be the dilation one parameter unitary group of $\A_0$ and $\B_0$, thus $U(s)\B_0(u)U(-s) = \B_0(e^s u)$ and $U(s)\A_0(v)U(-s) = \A_0(e^s v)$. Then
\begin{multline*}
\quad U(s)\A(u,v)U(-s) = U(s)\big(\B_0(u)\vee \A_0(R^2/v)\big)U(-s) =
\\ = \B_0(e^s u)\vee \A_0(e^s R^2/v) = \A(e^s u, e^{-s}v) =
\A\big(\Lambda(s)(u,v)\big) 
\qquad \end{multline*}
as desired.
\endproof
\begin{proposition} 
$\A(\LH)$ is generated by $\B_0(\RR^+)$ and $\A_0(\RR^+)$. So
$\A(\LH)=\A_0(\RR^+)$ if $\B_0(\RR^+)\subset\A_0(\RR^+)$, in particular  if duality for half-lines holds for $\A_0$.
\end{proposition}
\begin{corollary}\label{vKMS}
If $\B_0(\RR^+)\subset\A_0(\RR^+)$, $\xi$ is a vacuum (resp.\ KMS) vector for $\A$ (w.r.t.\ the boost) if it is a vacuum (resp.\ KMS) vector for $\A_0$.
\end{corollary}
\subsection{Constructing QFT on $\LH$ by an element of the semigroup $\Ed(\A)$}
\label{Vconstruction}
Let now $\A_0$ be a local, dilation covariant net of Neumann algebras on $\mathbb R^+$. We denote by 
$\Ed(\A_0)$ the group of unitaries on $\H$, commuting with the
dilation unitary group $U$, such that $V\A_0(1,\infty)V^*\subset \A_0(1,\infty)$. As $V$ commutes with dilations, $V\A_0(a,\infty)V^*\subset \A_0(a,\infty)$ for all $a>0$.
\begin{proposition}
Let $\A_0$ be in a KMS representation with an extremal KMS state. Then
for each $V\in\Ed(\A_0)$, one has  $V\A_0(\RR^+)V^*=\A_0(\RR^+)$.
\end{proposition}
\proof 
If the representation arises from an extremal KMS state, $\A_0(\RR^+)$
is a factor. In particular, the GNS vector $\xi$ is (up to phases) the
unique dilation invariant vector. Because $V$ commutes with the dilations, 
$V^*\xi$ is equal to $\xi$ up to a phase, therefore cyclic for $\A_0(\RR^+)$. Thus, $V\A_0(\RR^+)V^*$ is contained in $\A_0(\RR^+)$, cyclic on the invariant vector $\xi$, and
globally invariant under the modular group of $\A_0(\RR^+)$ (the
rescaled dilations). So $V\A_0(\RR^+)V^*=\A_0(\RR^+)$ by Takesaki's theorem.  
\endproof
If $\A_0$ is a local, M\"obius covariant net of Neumann algebras on $\mathbb R$, then the restriction $\A_0|_{\RR^+}$ of $\A_0$ to $\RR^+$ is a local, dilation covariant net on $\mathbb R^+$ (in a KMS representation) and we simply set $\Ed(\A_0)= \Ed(\A_0|_{\RR^+})$.

Setting
\begin{equation}
\B_0(I) = V\A_0(I)V^*,\quad I\in\I_+\ ,
\end{equation}
$\B_0$ is a local net of von~Neumann algebras on $\RR_+$, with dilation unitary group $U$. The net $\B_0$ is forwardly local w.r.t.\ $\A_0$: if $0<a<b<c<d$ we have $\B_0(c,d)\subset \B_0(c,\infty)\subset \A_0(c,\infty)$ so $\B_0(c,d)$ commutes with $\A_0(a,b)$.

With $V$ a unitary in $\Ed(\A_0)$ we denote by ${\A_0}_V$ the local net on $\LH$ associated with $\A_0$ and $\B_0= V\A_0 V^*$, i.e.,
\[
{\A_0}_V(\O) = V\A_0(I_1)V^*\vee \A_0(I_2)\ ,
\]
where $\O=I_1\times I_2\in\K$.
So we have:
\begin{proposition}\label{AV}
${\A_0}_V$ is a local, boost covariant net of von~Neumann algebras on $\LH$.
\end{proposition}
\noindent
By Cor.\ \ref{vKMS}, ${\A_0}_V$ is in a vacuum (resp.\ KMS) representation if $\A_0$ is in a vacuum (resp.\ KMS) representation.

So, in particular, given a local M\"obius covariant net $\A$ on $\mathbb
R$, we have maps:
\begin{equation}\label{map0}
V\in\E(\A_0)\  \mapsto\  \text{BQFT net  ${\A_0}_V$ on $\LH$ in a ground representation}
\end{equation}
and
\begin{equation}\label{map}
V\in\Ed(\A_0)\  \mapsto\  \text{BQFT net  ${\A_0}_V$ on $\LH$ in a KMS representation}.
\end{equation}
where the inverse temperature is $\beta = 2\pi$. Here, as in \cite{LW}, $\E(\A_0)$ is the semigroup of unitaries commuting with translations such that $V\A_0(\RR^+)V^*\subset V\A_0(\RR^+)V^*$.

We shall say that two nets $\B_1$, $\B_2$ on $\LH$, acting on the Hilbert spaces $\H_1$ and $\H_2$, are \emph{locally isomorphic} if for every double cone $\O\in\K$  there is an isomorphism $\Phi_\O:\B_1(\O)\to\B_2(\O)$ such that
\[
\Phi_{\tilde\O} |_{\B_1(\O)} = \Phi_\O 
\]
if $\O,\tilde\O\in\K$, $\O\subset \tilde\O$ and 
\[
 U_2(t)\Phi_\O(X) U_2(-t) =\Phi_{\O+t}(U_1(t)XU_1(-t)),\quad X\in \B_1(\O)\ ,
\]
with $U_1$ and $U_2$ the corresponding boost unitary groups on $\H_1$ and $\H_2$.
\begin{proposition}
\label{lociso}
Let $\A_0$ be a local M\"obius covariant net of Neumann algebras on $\mathbb R$ with the split property. If $V$ and $W$ are unitaries in $\Ed(\A)$ the nets ${\A_0}_{V}$ and ${\A_0}_{W}$ on $\LH$ are locally isomorphic.
\end{proposition}
\begin{proof} 
The proof is similar to the one for the case of Boundary QFT on the
half-space \cite{LW}.
\end{proof}
\subsection{Induced nets on $\LH$}
\label{Induced}
The logarithmic map $t+x:=\log u$, $t-x:=\log R^2/v$ is a diffeomorphism
of the hyperboloid $\LH$ to the Minkowski half-space
$M_+=\{(t,x)\in\RR^2:x>0\}$. We write $I\times J:=\{(t,x):t+x\in
I,t-x\in J\}$. Then $I\times J\subset M_+$ iff $I>J$. The
diffeomorphism identifies double cones $I\times J\subset M_+$ with
double cones $e^I\times R^2e^{-J}\subset \LH$.

Let $\A(I)=\A_0(e^I)$ and $\B(I)=\B_0(e^I)$ be the translation
covariant nets on $\RR$ associated with $\A_0|_{\RR^+}$ and
$\B_0|_{\RR^+}$ by the logarithmic map. We use the diffeomorphism to
transfer the above net \eqref{bqft} on $\LH$ to a net on $M_+$, namely
$\A_+(I\times J):= \A_{\LH}(e^I\times R^2e^{-J}) = \B_0(e^I)\vee
\A_0(e^J)$ for $I>J$. Thus,
\begin{equation}\label{chiralbcft}
\A_+(I\times J) = \B(I)\vee \A(J).
\end{equation}
This is a local net on the Minkowski half-space
$M_+$. If $\A$ and $\B$ are the same (M\"obius covariant)
local nets, \eqref{chiralbcft} is the BCFT net of chiral observables,
constructed in \cite{LR1}. 

The discussion in the beginning of this
section shows that \eqref{chiralbcft} is local if $\A$ and $\B$ not
necessarily coincide, but $\B$ is forwardly local w.r.t.\ $\A$. 
Indeed, if $I_1\times J_1\subset M_+$ and $I_2\times J_2\subset M_+$
are spacelike to each other, then without loss of generality
$I_2>I_1>J_1>J_2$, so that forwardly locality is necessary and sufficient to
establish locality of the net $\A_+$. 

In \cite{LW}, local nets $\B$ that are forwardly local w.r.t.\ a given
local net $\A$ are constructed by conjugation with a unitary from the
semigroup $\E$. Thus, the present construction comprises the two
previous cases as special cases (after mapping the hyperboloid onto
the half-space); but the emphasis in this article is on the more 
flexible choice of states and representations. In particular, we here
admit $\A_0$ in the vacuum representation, which means that
$\A=\A_0|_{\RR_+}$ is in a KMS representation (see above).

In this subsection, we want to place the present construction into
context with the more general construction of non-chiral BCFT
observables on the half-space, presented in \cite{LR1}. Namely, in
\cite{LR1}, we have considered a nonlocal but relatively local chiral
extension of $\A$, which we call $\C$ here to prevent
confusion with $\B$ above; i.e., $\A(I)\subset \C(I)$ for all $I\in\I$. We then define the {\em induced net} on $M_+$
\begin{equation}\label{bcft}
\C_+(I\times J):= \C(K)'\cap \C(L),
\end{equation}
where $K$ is the open interval between $I$ and $J$, and $L$ is the
open interval whose closure equals the closure of $I\cup K\cup
J$ (i.e., $L=I\cup K\cup J$ plus two interior points). $\C_+$ is a
local net on $M_+$, and it contains the subnet $\A_+(I\times
J)=\A(I)\vee \A(J)$ of chiral observables. 

Also here, the conditions in \cite{LR1} can actually be relaxed:
Let $\A$, $\B$ be two local nets of von~Neumann algebras on $\RR$,
and $\C$ a possibly nonlocal net that extends both $\A$ and $\B$. 
\begin{proposition}\label{inducednet}
If $\B$ is forwardly local w.r.t.\ $\A$, and $\C$ is forwardly
local w.r.t.\ $\A$ and backwardly local w.r.t.\ $\B$, then the induced net
$\C_+$ on the half-space $M_+$ defined by \eqref{bcft} is local
and contains $\A_+(I\times J) = \B(I)\vee \A(J)$. Moreover, $\C_+$
is covariant under time translations if $\A$, $\B$ and $\C$ are
translation covariant and the translations of $\C$ restrict to the
translations of $\A$ and $\B$. 
\end{proposition}
\noindent
There is also a partial converse: given a local net $\D_+$ on $M_+$,
one may define associated nets on the boundary $\RR$ as follows: 
\bea
\A(J) &:=& \bigcap_{I:\, I>J} \D_+(I\times J) \notag \\
\B(I) &:=& \bigcap_{J:\, J<I} \D_+(I\times J) \notag \\ 
\label{generated}
\C(L) &:=& \bigvee_{I,J\subset L:\, I>J} \D_+(I\times J).
\ea
\noindent
The following result is also straightforward.
\begin{proposition}\label{boundarynets}
$\C$ contains both $\A$ and $\B$, and $\C$ is forwardly local w.r.t.\
$\A$ and backwardly local w.r.t.\ $\B$. In particular, $\B$ is
forwardly localy w.r.t.\ $\A$, and both $\A$ and
$\B$ are local. If in addition 
$\D_+$ is time translation covariant, then $\A$, $\B$ and $\C$
are translation covariant and the translations of $\C$ restrict to the
translations of $\A$ and $\B$. Moreover, one has the inclusions 
$$\B(I)\vee A(J) \subset \D_+(I\times J) \subset \C(K)'\cap \C(L).$$
\end{proposition}
\noindent
In particular, every local net on $M_+$ is intermediate between a net
of the form \eqref{chiralbcft} and an induced net of the form \eqref{bcft}. 
\begin{proof} 
Only the locality properties need a little argument. If $L > J$,
choose $I>L$. Then $I\times J$ is spacelike separated from $I_1\times
J_1$ for all $I_1,J_1\subset L$, $I_1>J_1$. Because $\D_+$ is local, it
follows that $\C(L)$ commutes with $\A(J)$, i.e., $\C$ is forwardly
local w.r.t.\ $\A$. Backward locality of $\C$ w.r.t.\ $\B$ is similar,
and the other statements follow because $\C(I)$ contains $\A(I)$ and
$\B(I)$. Notice that a net is local if it is forwardly or backwardly
local w.r.t.\ itself.
\end{proof}
It would be interesting to characterize possible nonlocal nets $\C$
for which the induced net is not trivial ($=\CC$). E.g., if $\D_+$ is
given as in \eqref{chiralbcft}, and $\C$ is generated from $\D_+$ as in
\eqref{generated}, then the net induced from $\C$ contains at least
$\B(I)\vee \A(J)$. Even if $\B=\A$, the extension $\A(I)\vee
\A(J)\subset \C_+(I\times J)$ is nontrivial in general, reflecting the
superselection sectors of $\A$ \cite{KLM,LR1}. Thus, for a given local
net $\A$ and an element $V$ of the associated semigroup $\E$, putting
$\B=V\A V^*$, it would be a highly interesting problem to understand
the structure of the resulting induced net in terms of $\A$ and $V$. 

The analogous question arises when $\A$ and $\B$ are diffent chiral
extensions of a given chiral net (e.g., that of the stress-energy
tensor), differing by a cohomological twist of the Q-system \cite{KL}. 
We hope to return to these issues in a future publication. 

Clearly, the entire discussion of this subsection can be transferred 
back to the hyperboloid by the change of coordinates as above. Thus, 
on the hyperboloid one has for $I,J\in \I_+$, $I>R^2J\inv$,
$$C_{\LH}(I\times J) = C_0(K)'\cap C_0(L),$$
where $C_0$ is a net on $\RR_+$, and $K$ is the open interval between 
$I$ and $R^2J\inv$, and $L=I\cup K\cup R^2J\inv$ plus two interior points. 
If $C_0$ is dilation covariant, then $C_{\LH}$ is Lorentz covariant.

\section{Endomorphisms of standard subspaces}
\subsection{Preliminary comments}
With $a\in (0,\infty]$ we  denote by $\mathbb S_{a}$ the strip of the complex plane $\{z\in\mathbb C: 0<\Im z < a \}$ (so $\mathbb S_{\infty}$ is the upper plane).
Let $f\in L^1(\RR)$, and $\f$ the Fourier transform of $f$. Then, if
$f$ is real and supp$(f)\subset \RR^+$, $\f$ is a symmetric function
($\f(-s) =\overline{\f(s)}$) and belongs to the Hardy space $\Hin$,
namely it is the boundary value of a bounded analytic function on the
upper half-plane ${\mathbb S}_{\infty}$. The space of such functions
is weakly dense in $\Hin$. 

Let $\H$ be a (complex) Hilbert space, $H_1$ a real Hilbert subspace of $\H$ and $K$ a selfadjoint operator on $\H$. Suppose that 
\[
e^{itK}H_1\subset H_1,\quad \forall t\geq 0 \ .
\]
For $f$ and its Fourier transform $\f$ as before, we define
\[
\f(K) = \int_{-\infty}^\infty f(t)e^{itK}dt\ .
\]
Then, if supp$(f)\subset \RR^+$, we have
$\f(K)H_1\subset H_1$. 

Suppose further that $K$ has Lebesgue spectrum. Then the map $\f\in L^\infty(\RR)\mapsto \f(K)\in B(\H)$ is weakly continuous so, taking limits, we have
\[
\f(K)H_1\subset H_1\ , \quad \forall \f\in\Hin ,\ \f\ \text{symmetric},
\]
and, in particular, every symmetric inner function $\f$ on $\Si$ gives a unitary $V=\f(K)$ such that $VH_1\subset H_1$. 

We shall see here below a situation where this occurs. In our case a converse will hold. Our standard subspace methods will be appropriate as they allow us to prove such a converse too; moreover the reducible case is treatable too. 

\subsection{Characterization of the semigroup}
Let $\H$ be a complex Hilbert space, $H\subset\H$ a 
{\em standard subspace}, i.e., a closed real subspace such that $H+iH$ is
dense in $\H$ and $H\cap iH=\{0\}$, and $\Delta_H$, $J_H$ the
modular operator and modular conjugation of $H$. 

Recall that, if $V\in B(\H)$ a bounded linear operator on $\H$, we have the following \cite{AZ,LN}. 
\begin{proposition}\label{AZL}
The following are equivalent:
\begin{itemize}
\item[$(i)$]
$VH\subset H$ 
\item[$(ii)$]
The map $s\in\RR \to V(s)\equiv\Delta^{-is}V\Delta^{is}$ extends to a bounded weakly continuous function on the closed strip $\overline{\mathbb S_{1/2}}$, analytic in $\mathbb S_{1/2}$, such that
$V(i/2) = JVJ$.
\end{itemize}
where $\Delta = \Delta_H, J =J_H$.
\end{proposition}
\noindent
Let now $(H,T)$ be a {\em standard pair} of the Hilbert space $\H$. 
Namely $H$ is a standard subspace of the Hilbert space $\H$ and there exists a one parameter unitary group $T(t) = e^{itP}$ on $\H$ such that $T(t)H \subset H$ for all $t\geq 0$, and $P>0$.

Assume that $(H,T)$ is irreducible and
let $H_a\equiv T(a)H$, so $\Delta_{H_1} = T(1)\Delta T(-1)$. Then
\begin{equation}\label{hsm}
\Delta^{-is}H_1 \subset H_1,\quad s\geq 0\ ,
\end{equation}
indeed $\Delta^{-is}H_1 = \Delta^{-is}T(1)H = T(e^{2\pi s})\Delta^{-is}H =
T(e^{2\pi s})H = H_{e^{2\pi s}}$. 

If an inclusion of standard subspaces $H_1\subset H$ satisfies the condition \eqref{hsm} we shall say that $H_1\subset H$ is a \emph{half-sided modular (hsm)} inclusion (of standard subspaces).
\begin{proposition}\label{equiv1}
Let $H_1\subset H$ be an inclusion of standard subspaces of $\H$. The following are equivalent:
\begin{itemize}
\item[$(i)$] $H_1\subset H$ is a half-sided modular inclusion.
\item[$(ii)$] There exists a standard pair $(H,T)$ with $T(1)H = H_1$.
Moreover $T$ is uniquely determined.
\end{itemize}
\end{proposition}
\proof
The equivalence between $(i)$ and $(ii)$ follows by the analogs of theorems by Borchers and Wiesbrock, see \cite{LN}, except for the uniqueness of $T$ that we show now.
Given the hsm inclusion $H_1\subset H$, let $T_1$ and $T_2$ be
one-parameter unitary groups on $\H$ such that $T_i(t)H\subset H$, $t\geq 0$ and $T_i(1)H = H_1$. Then $T_2(-1)T_1(1)H=H$, so $T_2(-1)T_1(1)$ commutes with $\Delta$. On the other hand $\Delta^{-is}T_2(-1)T_1(1)\Delta^{is} = T_2(-e^{2\pi s})T_1(e^{2\pi s})$, so 
$T_2(-t)T_1(t)= T_2(-1)T_1(1)$, $t\geq 0$; thus $T_2(1-t) = T_1(1-t)$ and $T_1 = T_2$ by the group property.
\endproof
\begin{proposition}
Let $H_1\subset H$ be an inclusion of standard subspaces of $\H$ and
$V$ be a unitary on $\H$ such that $VH_1\subset H_1$.
Then $V$ commutes with $\Delta$ iff $VH=H$\!\! .
\end{proposition}
\proof
If $VH=H$ then $V$ commutes with $\Delta$ by modular theory, see \cite{LN}. 
We show the converse.

As $\Delta^{is}H_1 = H_{e^{-2\pi s}}$ and $V$ commutes with $\Delta$, we have $VH_a\subset H_a$ for all $a>0$, thus $VH\subset H$ because $H=\overline{\cup_{a>0}H_a}$ (as $\overline{\cup_{a>0}H_a}$ is a standard $\Delta^{is}$-invariant subspace of $H$).

Now $VH\subset H$ and $V\Delta^{is}= \Delta^{is}V$ imply that $VH=H$ because $VH$ is a standard, $\Delta^{is}$-invariant subspace of $H$.
\endproof
\begin{lemma}\label{lemV}
Let $H_1\subset H$ be an inclusion of standard subspaces of $\H$ as in Prop. \ref{equiv1}. Let $V$ be a unitary on $\H$, commuting with $\Delta$.
Then $V=\f(K)$ where $2\pi K=-\log\Delta$ and $\f$ is a Borel function on $\RR$ with $|\f(x)| = 1$ for almost all $x\in\RR$.

Moreover $VH_1\subset H_1$ iff the operator-valued function 
\[
F:a\in(1,\infty)\mapsto \f(K + aP)
\]
admits a bounded analytic continuation in the upper half-plane
$\Si$, and $F(z)|_{z=0} =  \f(K)$.
\end{lemma}
\proof
As $(H,T)$ is irreducible and $V$ commutes with $\Delta$, we have $V=\f(K)$ with $\f$ is a Borel function on $\RR$ and $|\f(x)| = 1$ for almost all $x\in\RR$ by the unitarity of $V$. 

By Prop. \ref{equiv1} $VH= H$ and in particular
\[
VJ = JV\ .
\]
So we have
\[
\f(K) = V= JVJ = J\f(K)J = \bar\f(-K)\ ,
\]
that is $\f(-x)=\bar\f(x)$ for almost all $x\in\RR$.
By the implication $(i)\Rightarrow (ii)$ in Prop. \ref{AZL} above we then have
\[
\Delta_{H_1}^{-is}V\Delta_{H_1}^{is}|_{s=i/2}= J_{H_1}VJ_{H_1}
\]
(namely the map $s\in\RR \to V_1(s)\equiv\Delta_{H_1}^{-is}V\Delta_{H_1}^{is}$ extends to a bounded weakly continuous function on the closed strip $\overline{\mathbb S_{1/2}}$, analytic in $\mathbb S_{1/2}$, such that
$V_1(i/2) = J_{H_1}VJ_{H_1}$). 

So
\[
T(1)\Delta^{-is}T(-1)VT(1)\Delta^{is}T(-1)|_{s=i/2}= T(1)JT(-1)VT(1)JT(-1)
\]
namely
\[
\Delta^{-is}T(-1)VT(1)\Delta^{is}|_{s=i/2}= JT(-1)VT(1)J\ .
\]
Since $JT(1)J = T(-1)$ and $JVJ = V$, we have
\[
T(-e^{2\pi s})VT(e^{2\pi s})|_{s=i/2}= T(1)VT(-1)
\]
or
\[
T(-(1+e^{2\pi s}))VT(1+e^{2\pi s})|_{s=i/2}= V \ .
\]
In other words
\[
T(-(1+e^{2\pi s}))\f(K)T(1+e^{2\pi s})|_{s=i/2}= \f(K)
\]
in the sense that the operator-valued function $F:a\in(1,\infty)\mapsto T(a)  \f(K) T(-a)$ admits a bounded analytic continuation in the upper half-plane $\Si$ and
\[
F(z)|_{z=0}=T(-z)\f(K)T(z)|_{z=0}= \f(K)\ .
\]
Now 
\[
T(a)  \f(K) T(-a)= \f(T(a) K T(-a)) 
\]
and
\[
T(a) \Delta^{-is} T(-a) = T(a) \Delta^{-is} T(-a) \Delta^{is} \Delta^{-is} = 
T(a)T(-a e^{2\pi s})\Delta^{-is}=
T(a(1- e^{2\pi s}))\Delta^{-is}
\]
so, differentiating at zero w.r.t.\ $s$ the first and the last member, we have
\[
T(a) K T(-a) =  K + aP\ .
\]
Thus, the operator-valued function $F:a\in(1,\infty)\mapsto 
\f(K + aP)$ admits a bounded analytic continuation in the upper half-plane
$\Si$, and 
\[
 F(z)|_{z=0}=\f(K + zP)|_{z=0} =  \f(K)\ .
\]
The above arguments are reversible so the Lemma is proved.
\endproof
\begin{theorem}\label{semichar}
Let $(H,T)$ be an irreducible standard pair of the Hilbert space $\H$, set $H_1 \equiv T(1)H$ and let $V$ be a unitary on $\H$ commuting with $\Delta_H$.

The following are equivalent:
\begin{itemize}
\item[$(i)$] $VH_1\subset H_1$\ ;
\item[$(ii)$] $V= \f(K)$ where $2\pi K=-\log\Delta_H$ and
$\f$ is a symmetric inner function on the upper half-plane.
\end{itemize}
Moreover the implication $(ii)\Rightarrow (i)$ is true also if the standard pair $(H,T)$ is reducible.
\end{theorem}
\proof
$(ii)\Rightarrow (i)$: If $\f$ is a symmetric inner function on the upper half-plane, then by Cor. \ref{analf} the operator-valued function $F:a\in(1,\infty)\mapsto 
\f(K + aP)$ admits a bounded analytic continuation in the upper half-plane
$\Si$, and $F(z)|_{z=0} =  \f(K)$. Therefore by Lemma \ref{lemV} $V= \f(K)$ satisfies $VH_1\subset H_1$.
 
The implication $(i)\Rightarrow (ii)$ is proved in Appendix \ref{endproof}.
\endproof
With $(H,T)$ a standard pair of $\H$, we shall denote by $\E_H(H_1)$ the semigroup of all unitaries on $\H$ such that $VH=H$ and $VH_1\subset H_1$. By Thm.\ \ref{semichar}, if $(H,T)$ is irreducible, $\E_H(H_1)$ is naturally isomorphic to the semigroup of symmetric inner functions on $\Si$\! .

\section{Classes of models}
\subsection{The two local nets on $\LH$ associated with a chiral net}
Let $\A$ be a local M\"obius covariant net of von~Neumann algebras on
$\RR$. There are two canonical local, boost covariant nets on $\LH$
associated with $\A$, namely the ones given by \eqref{map0} and
\eqref{map} with $V=1$. 

In other words, $\A$ gives rise to two local dilation covariant nets
on $\RR^+$. One is $\A|_{\RR^+}$: the associated net on $\LH$ is in a
KMS state at inverse temperature $\beta = 2\pi$. The second one is the
dilation covariant net obtained by the translation covariant net $\A$
on $\RR$ by the logarithmic change of variable: the associated net on $\LH$ is in a ground state.
\subsection{The net associated with a KMS state on a chiral net}
Let $\A$ be a local translation covariant net of von~Neumann algebras on $\RR$. Every KMS state $\omega$ on $\A$, namely every locally normal KMS state at inverse temperature $\beta$ w.r.t.\ translations on the C$^*$-algebra $\overline{\cup_{a>0}\A(-a,a)}$, gives rise to a local, boost covariant $\A_\omega$ net on $\LH$ by be construction in Section \ref{constr1}. 

Assuming half-line duality of $\A$, the net $\A_\omega$ is in a KMS state at inverse temperature $\beta$.

The KMS states for the $U(1)$-current nets are known, see \cite{CLTW2}, so we have an infinite family of nets on $\LH$ in a KMS state.

An infinite (possibly complete) family of KMS states for the Virasoro nets is also given in
\cite{CLTW2}, providing another infinite family of nets on $\LH$ in a KMS state.
\subsection{The semigroup and family of models associated with the $U(1)$-current}
Let $\A^{(0)}$ be the M\"obius covariant net on $\RR$ associated with the $U(1)$-current $j$, and $\A^{(k)}$ the net generated by the $k$-derivative of $j$. 

With $V_0$ a unitary on the one-particle Hilbert space $\H_0$ we denote by $\Gamma(V_0)$ its second quantization promotion to the Bosonic Fock space over $\H_0$.
We shall refer to a unitary of the form $\Gamma(V_0)$ as a \emph{second quantization unitary}.
Similarly as in the half-space case \cite[Thm.\ 3.6]{LW}, we then have:
\begin{theorem} A second quantization unitary $\Gamma(V_0)$ belongs to
  $\Ed({\A^{(k)}})$ if and only if $V_0=\f(K^{(k)})$ with $\f$ is the
  boundary value of a symmetric inner function on $\Si$. Here
  $K^{(k)}$ is the generator of the dilation unitary group on the
  one-particle Hilbert space $\H_0^{(k)}$.
\end{theorem}
\noindent
Therefore, by Cor.\ \ref{vKMS}, we have:
\begin{corollary} For every symmetric inner function $\f$ on $\Si$ there is a QFT local net ${\A_V^{(k)}}$ of von~Neumann algebras on $\LH$ in a KMS representation at $\beta = 2\pi$ ($V = \Gamma(\f(K^{(k)}))$).
\end{corollary}

\bigskip


\appendix
\noindent{\LARGE \bf Appendix}

\section{On the $H^{\infty}$ functional calculus of Sz.-Nagy and Foias}
Let $\H$ be a Hilbert space and $T\in B(\H)$. If $||T||\leq 1$, $T$ is called a contraction. If $T$ is a contraction and there is no direct sum decomposition $T = T_1\oplus T_2$ with $T_2$ unitary, one says that $T$ is completely non-unitary.

If $T$ is a completely non-unitary contraction, there is an $H^{\infty}$ functional calculus for $T$, i.e., a Banach algebra homomorphism
\[
\f\in \mathbb H^{\infty}(\mathbb D)\mapsto \f(T)\in B(\H) \ .
\]
defined by
\[
\f(T) = \lim_{r\to 1^-}\f_r(T)\ ,
\]
where $\f_r(z) = \f(rz)$ (note that $\f_r$ is analytic in the disk of radius $1/r$ so the usual holomorphic functional calculus applies because sp$(T)\subset \overline{\mathbb D}$).
\begin{proposition}\label{hol1}
Let $T(z)$ be an operator-valued analytic function on a region $\cal G$ of $\mathbb C$, with $T(z)$ completely non-unitary contractions on $\H$. Given $\f\in \mathbb H^{\infty}(\mathbb D)$, the function $z\in{\cal G} \mapsto \f(T(z))$ is analytic in $\cal G$.
\end{proposition}
\proof
With $0<r<1$, the function $z\in{\cal G} \mapsto \f_r(T(z))$ is analytic in $\cal G$ by the usual holomorphic functional calculus. As $r\to 1^-$, $\f_r$ converges to $\f$ uniformly on $\mathbb D$, so $\f_r(T(z))$ converges to $\f(T(z))$  uniformly on $\mathbb D$ too by the continuity of the $H^{\infty}$ functional calculus.
\endproof
A densely defined linear operator $A: D(A)\subset \H\to \H$ is accretive if $\Re(\xi, A\xi)\geq 0$ for all $\xi \in D(A)$ and maximal accretive if there is on non-trivial accretive extension of $T$ on $\H$.

A maximal accretive operator is closed and its spectrum is contained
in the right half-plane $\Re z\geq 0$. In this case the Cayley
transform $T=(A+1)(A-1)^{-1}$ is a contraction. Suppose that $T$ is completely non-unitary. Then one can define a functional calculus for $T$
\[
\f(A) = \f_0(h^{-1}(T))
\]
for every function $\f$ in $\mathbb H^{\infty}(\Re z > 0)$, where
$h(\l)= \frac{\l+1}{\l-1}$ is the Cayley map (because
$\f_0 =\f\circ h\in\mathbb H^{\infty}(\mathbb D)$).

\subsection{Case of $K+zP$}
Let $U$ be the irreducible, positive energy unitary representation of the ``$ax + b$'' group with non-trivial translation group. Let $P$ and $K$ be the generators of the translation and dilation group.  The operator $-i(K +  zP)$ is densely defined on $D=D(P)\cap D(K)$. If $\Im z \geq 0$ the operator $-i(K +  zP)$ is accretive: with $z = a+ib$ and $\xi\in D$
\[
\Re(\xi, -i(K + z P)\xi) = \Re [-i(\xi, K\xi) -  ia(\xi, P\xi) +  b(\xi, P\xi)]
= b (\xi, P\xi)\geq 0 \ .
\]
Denote by $A(z)$ the closure of $-i(K + zP)$ ($\Im z\geq 0$). 
\begin{lemma} For every $z\in \mathbb C$ with $\Re z\geq0$:
\begin{itemize}
\item $-iA(z)$ is maximal accretive.
\item The Cayley transform $T(z)=(A(z)+i)/(A(z)-i)$ is a completely non-unitary contraction.
\item The operator valued function $z\mapsto T(z)$ is analytic on $\Si$ and continuous on $\overline{\Si}$.
\end{itemize}
\end{lemma}
\proof
To show the first part, note that $-iA(z)$ is accretive, so it is maximal accretive iff the range of $-i(K +  zP)+i$ is dense if $\Im z>0$, so we have to show that the range of $K +\lambda P -I$ is dense if $\Re\lambda <0$. We may assume that $U$ is in the Schr\"odinger representation, namely $\H = L^2(\RR)$, $K = i\frac{d}{dx}$ and $P$ is the multiplication by $e^x$. Then the result follows by elementary first order linear differential equation theory.

The second and the third part then follow, see \cite{NF}.
\endproof
Therefore, if $\Im z\geq 0$, we can define as above $\f(A(z))$ for every $\f\in \mathbb H^{\infty}(\Si)$, and we have:
\begin{corollary}\label{analf}
For every fixed $\f\in \mathbb H^{\infty}(\Si)$ the map $z\mapsto \f(A(z))$ is analytic on $\Si$ and bounded, continuous on $\overline{\Si}$.
\end{corollary}

\section{End of proof of Theorem \ref{semichar}}\label{endproof}

\begin{proposition} \label{appb}
With the notions in Thm.\ \ref{semichar}, if $V$ is a unitary commuting with dilations such that $VH_1\subset H_1$, then $V=\f(K)$ with $\f\in\Hin$.
\end{proposition}

\noindent
As already noted, since $(H,T)$ is assumed to be irreducible, $V$ must be of the form $V=\f(K)$ with $\f\in L^{\infty}(\RR)$, and we must show that $\f\in\Hin$. We prove this in a few steps. 
\medskip

\noindent $\bullet$ {\it We may assume that $\f$ rapidly decreases at $\pm\infty$.}
\smallskip

\noindent
Let $j_n$ be non-negative, smooth function on $\RR$ with integral 1 and supp$(j_n)\subset [0,1/n]$, $n\in\mathbb N$. 
Then the sequence $\{j_n\}$ is an approximate unit in $L^1(\RR)$ and $\hat j_n$ converges weakly to $1$ in $L^\infty(\RR)$ (in the $\sigma(L^\infty,L^1)$-topology). Set $\f_n\equiv \hat
j_n \f$; then each $\f_n$ rapidly decreases at $\pm\infty$ and their
sequence converges weakly to $\f$ in $L^1(\RR)$. Morever $\hat j_n\in\Hin$, so $\hat j_n(K)H_1\subset H_1$ by the proven implication $(ii)\Rightarrow(i)$ in Thm.\ \ref{semichar}. So $\f_n(K)H_1\subset H_1$.  

It is therefore sufficient to prove the claim $\f_n\in\Hin$ for
rapidly decreasing functions $\f_n$. Then also $\f\in\Hin$ follows,
because $\f_n\to \f$ and $\Hin$ is weakly closed in $L^\infty(\RR)$. 
 
\medskip

\noindent $\bullet$ {\it There is a dense linear space $\D\subset\H$ such that:

$(a_1)$ For every fixed $t>0$,
the function $a\in\RR\mapsto (\eta,e^{it(K+aP)}\xi)$ is the boundary value of a function in $\Hin$, continuous in $\overline{\Si}$;

$(a_2)$ For every fixed $z$ with $\Im z> 0$, the function
$t\in\RR^+\mapsto (\eta,e^{it(K+zP)}\xi)$ belongs to
$L^2$. }
\smallskip

\noindent
Since, for a fixed $t>0$, the function $\f_t: s\mapsto e^{its}$ is a symmetric inner function on $\Si$, by Lemma\ \ref{lemV} the operator-valued map $a\in\RR\mapsto e^{it(K+aP)}$ is bounded analytic on $\Si$, continuous in $\overline{\Si}$, because $e^{itK}H_1\subset H_1$. So the first statement $(a_1)$ follows.

To show $(a_2)$, we may work in the Schr\"odinger representation, so $\H=L^2(\RR, dx)$, $e^{itK}$ is the translation group and $P$ is the multiplication by $e^x$. Then
\begin{multline*}
e^{it(K+aP)} = T(a)e^{itK}T(-a) :\\
 \xi(x)\mapsto T(a)e^{itK}e^{-iae^x}\xi(x)
= T(a) e^{-iae^{x-t}}\xi(x-t) = e^{iae^x}e^{-iae^{x-t}}\xi(x-t) = e^{iae^x(1 - e^{-t})}\xi(x-t),
\end{multline*}
so
\[
(\eta,e^{it(K+aP)}\xi) = \int e^{iae^x(1 - e^{-t})}\xi(x-t)\overline{\eta(x)}dx \ .
\]
If $\xi,\eta$ are rapidly decreasing functions we then have
\[
|(\eta,e^{it(K+zP)}\xi)|
\leq \int |e^{ize^x(1 - e^{-t})}\xi(x-t)\eta(x)|dx
= \int e^{-be^x(1 - e^{-t})}|\xi(x-t)\eta(x)|dx
\leq |\xi|*|\tilde\eta|(t)
\]
where $z = a+ ib$, $b >0$, $\tilde\eta(t)=\eta(-t)$, and $*$ denotes the convolution product. The statement follows because $|\xi|*|\tilde\eta|$ is rapidly decreasing.
\smallskip

\noindent
In the following $\f\in L^{\infty}(\RR)$ is rapidly decreasing at infinity, $V= \f(K)$ maps $H_1$ into itself.
\medskip

\noindent $\bullet$ {\it Let $f$ be the Fourier anti-transform of $\f$. Then $f\in L^2(\RR)$ and
\[
(\eta,\f(K)\xi) = \int_{-\infty}^\infty f(t)(\eta,e^{itK}\xi)dt\ 
\]
for all vectors $\xi,\eta$ such that $t\mapsto (\eta,e^{itK}\xi)$ belongs to $L^2$, in particular for $\xi,\eta\in\D$.}
\smallskip

\noindent
This is obvious for all $f\in L^1$, and true in our case by an approximation argument.
\medskip

\noindent $\bullet$ {\it The map
\[
F_+: a\in\RR\mapsto \int_0^\infty f(t)(\eta,e^{it(K+aP)}\xi)dt
\]
admits a bounded analytic continuation in $\Si$ for all $\xi,\eta\in\D$.
}
\smallskip

\noindent
By the above point we may define
\[
F_+(z) = \int_0^\infty f(t)(\eta,e^{it(K+zP)}\xi)dt
\]
for $\Im z>0$ providing the desired analytic continuation.
\medskip

\noindent $\bullet$ {\it The map
\[
F_-: a\in\RR\mapsto \int_{-\infty}^0 f(t)(\eta,e^{it(K+aP)}\xi)dt
\]
admits a bounded analytic continuation in $\Si$ for all $\xi,\eta\in\D$.
}
\smallskip

\noindent
This is now immediate because $F_- = F - F_+$ where $F(a)=
(\eta,\f(K+aP)\xi)$ admits an analytic continuation in $\Si$ because
$\f(K)H_1\subset H_1$. 
\medskip

\noindent $\bullet$ {\it Conclusion.
}
\smallskip

\noindent
We have 
\[
F_-(a)= \int_{-\infty}^0 f(t)(\eta,e^{it(K+aP)}\xi)dt = 
-\int_0^{\infty} f(-t)(\eta,e^{-it(K+aP)}\xi)dt
\]
and by the above argument $F_-$ admits a bounded analytic continuation in the lower half-plane $-\Si$. So $F_-$ is constant by Liouville theorem. So $F= F_+$ plus a constant that must be zero by $L^2$-integrability.
This means that supp$(f)\subset [0,\infty)$, so $\f\in\Hin$.

This completes the proof of Prop.\ \ref{appb}, and thus the
conclusion $(i)\Rightarrow (ii)$ in Thm.\ \ref{semichar}.

\bigskip

\noindent
{\bf Acknowledgements.} R.L. is grateful to E. Witten for stimulating comments and 
D. Voiculescu for pointing out ref.\ \cite{NF}.

\end{document}